\documentclass[a4paper,11pt]{article}
\pdfoutput=1 

\usepackage{jheppub} 

\usepackage[T1]{fontenc} 

\usepackage{amsthm}
\usepackage{amssymb}   
\usepackage{mathtools}
\usepackage{accents}
\usepackage{hyperref}

\usepackage{natbib}
\title{\boldmath Canonical analysis of $BF$ gravity in $n$ dimensions}

\author[a]{Mariano Celada,}
\author[b]{Ricardo Escobedo}
\author[b,1]{and Merced Montesinos\note{Corresponding author.}}

\affiliation[a]{Centro de Ciencias Matem\'{a}ticas, Universidad Nacional Aut\'{o}noma de M\'{e}xico,\\
	UNAM-Campus Morelia, Apartado Postal 61-3, Morelia, Michoac\'{a}n 58090, Mexico}
\affiliation[b]{Departamento de F\'{i}sica, Cinvestav,\\ Avenida Instituto Polit\'{e}cnico Nacional 2508,
	San Pedro Zacatenco, 07360 Gustavo A. Madero, Ciudad de M\'exico, Mexico}

\emailAdd{mcelada@matmor.unam.mx}
\emailAdd{rescobedo@fis.cinvestav.mx}
\emailAdd{merced@fis.cinvestav.mx}

\abstract{In this paper we perform in a manifestly $SO(n-1,1)$ [or, alternatively $SO(n)$] covariant fashion, the canonical analysis of general relativity in $n$ dimensions written as a constrained $BF$ theory. Since the Lagrangian action of the theory can be written in two classically equivalent ways, we analyze each case separately. We show that for either action the canonical analysis can be accomplished without introducing second-class constraints during the whole process. Furthermore, in each case the resulting Hamiltonian formulation is the same as the canonical formulation with only first-class constraints recently obtained in Ref.~\cite{PalatininD} from the $n$-dimensional Palatini action.}

\newcommand{\uac}[1]{\underaccent{\tilde}{#1}}
\newtheorem{theorem}{Theorem}
\newtheorem{lemma}[theorem]{Lemma}

\begin{document} 
\maketitle
\flushbottom

\section{\label{sec:Intro} Introduction}
General relativity in dimension $n>4$ with vanishing cosmological constant can be reformulated as a constrained $BF$ theory~\cite{Freid_Puz}. This formulation generalizes the very well-known four-dimensional case to higher dimensions, which plays a fundamental role in the path integral quantization of gravity giving rise to the spinfoam models for gravity~\cite{perez2013,rovelli2014covariant}. In turn, spinfoam models for gravity are thought of as the covariant version of loop quantum gravity~\cite{RovBook,ThieBook}, which follows the canonical quantization program. Since both quantum approaches supplement each other, the knowledge of the canonical structure of the $BF$ models of gravity can help to establish connections between them~\cite{alexandrov2012spin}. Because of this, we are particularly interested in the canonical descriptions of $BF$ gravity (see Ref.~\cite{cqgrevBF} for a review on the subject).

Recently, a manifestly Lorentz-covariant canonical formulation for four-dimensional general relativity with the Immirzi parameter was reported in Ref.~\cite{Mont_Rom_Cel_2020} without having to introduce second-class constraints during the entire Hamiltonian analysis (see Refs.~\cite{Montesinos1801, MontRomEscCel} for a derivation with second-class constraints). The same formulation was also obtained from the canonical analysis of $BF$ gravity with the Immirzi parameter while again avoiding the presence of second-class constraints~\cite{BFNoSCC}, thus simplifying previous analyses (see for instance Ref.~\cite{CelMontcqg2920}).

On the other hand, it was also recently carried out the manifestly Lorentz-covariant canonical analysis with no second-class constraints of the $n$-dimensional Palatini action~\cite{PalatininD}. Such a canonical analysis had been previously reported after explicitly solving the second-class constraints that arise in the usual canonical analysis of the $n$-dimensional Palatini action~\cite{Bodendorfer_2013}. Furthermore, the canonical analysis of the $BF$-type action of Ref.~\cite{Freid_Puz} has also been performed, but introducing second-class constraints too~\cite{Thiemann_2013}. 

Since the action of Ref.~\cite{Freid_Puz} is classically equivalent to the $n$-dimensional Palatini action, and given that the canonical analysis of the latter action can be achieved without bringing in second-class constraints, it should be feasible to carry out the canonical analysis of the former action without their presence as well. In this paper we show that it is indeed possible, and we establish that the results obtained in Ref.~\cite{PalatininD} (with a vanishing cosmological constant) can also be reached from the canonical analysis of the $BF$-type action of Ref.~\cite{Freid_Puz} without having to introduce second-class constraints during the whole Hamiltonian analysis. Therefore, the current work can be regarded as the generalization to $n$ dimensions of the results of Ref.~\cite{BFNoSCC}.

Because the Lagrangian action of Ref.~\cite{Freid_Puz} can be expressed in two classically equivalent forms, in what follows we study both actions separately. As expected, the same canonical formulation emerges from them both, although we follow different paths in each case.

\paragraph*{Conventions.} Let us consider a principal bundle over a spacetime manifold $M$ of dimension $n$ ($n\geq$ 4). We assume the structure group is either $SO(n-1,1)$ [or, $SO(n)$, depending of the signature involved]. To perform the canonical analysis, we foliate the manifold $M$ by hypersurfaces $\Sigma_t$ homeomorphic to an $(n-1)$-dimensional manifold $\Sigma$ for all $t\in \mathbb{I}\subset\mathbb{R}$, such that $M$ has the global topology $\mathbb{I}\times\Sigma$ and $\Sigma$ has no boundary, $\partial\Sigma=\varnothing$. Points on $M$ are labeled with coordinates $x^{\mu}=(t,x^a)$, where $x^a$ ($a,b,\ldots$ take the values $1,\ldots,n-1$) are coordinates on $\Sigma$ ($t=constant$ specifies $\Sigma_t\sim\Sigma$). Spacetime indices are designated by greek letters $\mu,\nu,\ldots=\left\{t,a\right\}$, where $a$ labels spatial components of tensors, whereas $t$ refers to the time component (we represent both the time component of tensors and the time coordinate by the same letter $t$). Internal indices $I,J,\ldots$ take the values $0,\ldots,n-1$ and are raised and lowered with the internal metric $(\eta_{IJ})=\text{diag}(\sigma,\underbrace{1,\ldots,1}_{n-1})$, where $\sigma=-1$ for $SO(n-1,1)$ and $\sigma=1$ for $SO(n)$.  For any kind of indices, we define sets of $n-4$ and $n-5$ totally antisymmetric indices respectively by $[A]:=[A_1 \ldots A_{n-4}]=A_1 \ldots A_{n-4}$ and $\langle A \rangle:=[A_1 \ldots A_{n-5}] = A_1 \ldots A_{n-5}$. The antisymmetrizer is defined by
\begin{eqnarray}
	V^{[A_1 \ldots A_k]}:=\frac{1}{k!}\sum_{P\in S_k}\text{sgn} (P)V^{A_{P(1)} \ldots A_{P(k)}},
\end{eqnarray}
where the sum is over all the elements of the permutation group $S_k$ of order $k$. The weight of tensor densities is either denoted with a tilde ``$\sim$'' when possible or explicitly mentioned somewhere else in the paper. The $SO(n-1,1)$ or $SO(n)$ totally antisymmetric tensor $\epsilon_{I_1 \ldots  I_n}$ is such that $\epsilon_{01\ldots n-1}=1$. Similarly, the totally antisymmetric tensor density of weight $1$ ($-1$) is denoted as $\tilde{\eta}^{\mu_{1} \ldots \mu_{n}}$ ($\underaccent{\tilde}{\eta}_{\mu_{1} \ldots \mu_{n}}$) and satisfies $\tilde{\eta}^{t 1 \ldots n-1}=1$ ($\underaccent{\tilde}{\eta}_{t 1 \ldots n-1}=1$).

\section{$BF$-type action for general relativity}\label{secBF1}
\subsection{The action}

General relativity in $n$ dimensions can be expressed as a constrained $BF$ theory~\cite{Freid_Puz} (recall that $BF$ theory by itself defines a topological field theory~\cite{horowitz1989,birmingham1991,CAICEDO,cattaneo2000}). The heart of the formulation is a quadratic constraint on the $B$ field, which allows us to recover the $SO(n-1,1)$ or $SO(n)$ frame and reduces the $BF$-type action to the Palatini (or Einstein-Cartan) action. As explained in Ref.~\cite{Freid_Puz}, there are two classically equivalent formulations of the theory: Both action principles include a term with a product of $B$'s that can be antisymetrized in either the spacetime or the internal indices. We shall consider below both actions, which although lead to the same canonical theory, follow a different path in each case. Let us first consider the $BF$-type action that involves the antisymetrization in the spacetime indices of the product of $B$'s, which is given by
	\begin{eqnarray}
		S[A,\tilde{B},\Phi,\tilde{\mu}]=&&\int_M d^nx\left(\tilde{B}^{\mu\nu IJ}F_{\mu\nu IJ}+\Phi^{[\alpha]IJKL}\uac{\eta}_{[\alpha]\mu\nu\lambda\rho}\tilde{B}^{\mu\nu}{}_{IJ}\tilde{B}^{\lambda\rho}{}_{KL}\right.\notag\\
		&&\left.+\tilde{\mu}_{[\alpha]}{}^{[M]}\epsilon_{[M]IJKL}\Phi^{[\alpha]IJKL}\right),\label{BFact1}
	\end{eqnarray}
where $\tilde{B}^{\mu\nu}_{IJ}$ are bivectors of weight  $1$ taking values in the algebra $\mathfrak{so}(n-1,1)$ or $\mathfrak{so}(n)$; $F_{\mu\nu}{}^{IJ}:=\partial_{\mu} A_{\nu}{}^{IJ} - \partial_{\nu} A_{\mu}{}^{IJ} + A_{\mu}{}^I{}_K A_{\nu}{}^{KJ} - A_{\nu}{}^I{}_K A_{\mu}{}^{KJ}$ is the curvature of the $SO(n-1,1)$ or $SO(n)$ connection $A_{\mu IJ}$; $\Phi^{[\alpha]IJKL}$ is a tensor in both spacetime and internal indices that  satisfies $\Phi^{[\alpha]IJKL}=-\Phi^{[\alpha]JIKL}=-\Phi^{[\alpha]IJLK}=\Phi^{[\alpha]KLIJ}$; and $\tilde{\mu}_{[\alpha]}{}^{[M]}$ is a tensor density of weight $1$. Both $\Phi^{[\alpha]IJKL}$ and $\tilde{\mu}_{[\alpha]}{}^{[M]}$ play the role of Lagrange multipliers. An important fact about this action is that it does not feature a spacetime metric at all, which manifests the background independence of the theory; the metric itself can be regarded as a derived object. The action~\eqref{BFact1} is not the one extensively used in Ref.~\cite{Freid_Puz}, but it can be roughly obtained from it by exchanging the roles of the spacetime and internal indices in the last two terms. Notice that the analog of the last term of~\eqref{BFact1}, which involves $\tilde{\mu}_{[\alpha]}{}^{[M]}$ and imposes an additional condition on $\Phi^{[\alpha]IJKL}$, is not explicitly exhibited in the action of Ref.~\cite{Freid_Puz}.

\subsection{Einstein equations}

The variation of the action~\eqref{BFact1} with respect to the independent variables leads to
\begin{subequations}
	\begin{eqnarray}
		&&\delta A:\  D_{\mu}\tilde{B}^{\mu\nu IJ}=0,\label{eqmot1}\\
		&&\delta\tilde{B}:\ F_{\mu\nu}{}^{IJ}+2\Phi^{[\alpha]IJKL}\uac{\eta}_{[\alpha]\mu\nu\lambda\rho}\tilde{B}^{\lambda\rho}{}_{KL}=0,\label{eqmot2}\\
		&&\delta \Phi:\ \uac{\eta}_{[\alpha]\mu\nu\lambda\rho}\tilde{B}^{\mu\nu}{}_{IJ}\tilde{B}^{\lambda\rho}{}_{KL}+\tilde{\mu}_{[\alpha]}{}^{[M]}\epsilon_{[M]IJKL}=0, \label{eqmot3}\\
		&&\delta\tilde{\mu}:\ \epsilon_{[M]IJKL}\Phi^{[\alpha]IJKL}=0, \label{eqmot4}
	\end{eqnarray}
\end{subequations}
where $D_{\mu}$ is the ``$\mu$'' component of the $SO(n-1,1)$ or $SO(n)$ covariant derivative, $D_{\mu}T^{IJ\ldots}:=\partial_{\mu} T^{IJ\ldots}+A_{\mu}{}^I{}_K T^{KJ\ldots}+A_{\mu}{}^J{}_K T^{IK\ldots}+ \cdots$.

It is instructive to see how these equations give rise to the dynamics of general relativity. Equation~\eqref{eqmot3} implies that there exists a nondegenerate frame $e^{\mu}{}_I$ such that~\cite{Freid_Puz}
\begin{equation}
	\tilde{B}^{\mu\nu}{}_{IJ}=e\ e^{[\mu}{}_I e^{\nu]}{}_J,\label{solB}
\end{equation}
where $e:=\det(e_{\mu}{}^I)$ for $e_{\mu}{}^I$ the inverse of $e^{\mu}{}_I$. With this at hand, the spacetime metric can be defined as $g_{\mu\nu}:=e_{\mu}{}^Ie_{\nu}{}^J\eta_{IJ}$. Replacing~\eqref{solB} in~\eqref{eqmot1}, we obtain that $T_{\mu\nu}{}^I:=D_{\mu}e_{\nu}{}^I-D_{\nu}e_{\mu}{}^I$ vanishes, that is, $A_{\mu IJ}$ is torsion-free. This, together with the compatibility of $A_{\mu IJ}$ with the internal metric $\eta_{IJ}$ [because $A_{\mu IJ}$ is an $SO(n-1,1)$ or $SO(n)$ connection], implies that $A_{\mu IJ}$ actually is the spin-connection compatible with $e_{\mu}{}^I$: $A_{\mu IJ}=\omega_{\mu IJ}(e)$. Hence, the curvature $F_{\mu\nu IJ}$ is the Riemann tensor $R^{\mu}{}_{\nu\lambda\sigma}$ of the Levi-Civita connection, $R_{\mu\nu\lambda\sigma}=F_{\lambda\sigma IJ}e_{\mu}{}^I e_{\nu}{}^J$. Bearing in mind these facts, using again~\eqref{solB}, and working in the orthonormal basis determined by $e^{\mu}{}_I$, we can rewrite~\eqref{eqmot2} as
\begin{equation}\label{RLC}
	R^{IJ}{}_{KL}=e^{\mu}{}_K e^{\nu}{}_L F_{\mu\nu}{}^{IJ}=-2\Phi^{[M]}{}^{PQIJ}\epsilon_{[M]PQKL}, 
\end{equation} 
with $\Phi^{[M]}{}^{IJPQ}:=e_{\alpha_1}{}^{M_1 } \cdots e_{\alpha_{n-4}}{}^{M_{n-4}} \Phi^{\alpha_1 \ldots \alpha_{n-4}}{}^{IJPQ}$. From this, the Ricci tensor $R_{IJ}:=R^K{}_{IKJ}$ reads
\begin{equation}\label{Ricci}
	R_{IJ}=-2\Phi^{[M]}{}^{PQK}{}_I\epsilon_{[M]PQKJ}.
\end{equation}
Using the identity
\begin{eqnarray}
	\Phi^{[M]}{}^{PQK}{}_I\epsilon_{[M]PQKJ} = \frac{1}{4}\eta_{IJ}\Phi^{[M]PQRS}\epsilon_{[M]PQRS}- \frac{(n-4)}{4}\Phi^{\langle M\rangle}{}_I{}^{PQRS}\epsilon_{\langle M \rangle JPQRS}, 
\end{eqnarray}
we can rewrite the Ricci tensor~\eqref{Ricci} as
\begin{eqnarray}
	R_{IJ}=-\frac{1}{2}\eta_{IJ}\Phi^{[M]PQRS}\epsilon_{[M]PQRS}+\frac{(n-4)}{2}\Phi^{\langle M\rangle}{}_I{}^{PQRS}\epsilon_{\langle M\rangle JPQRS}.\label{Ricci2}
\end{eqnarray}
Finally, multiplying~\eqref{eqmot4} by the appropriate products of $e_{\alpha}{}^{M}$ to convert the spacetime indices into internal ones, we conclude that both terms on the right-hand side of~\eqref{Ricci2}  vanish, which means that $R_{IJ}=0$. These are the Einstein equations without a cosmological constant in vacuum. Therefore, the Einstein equations follow as stationary points of the action principle~\eqref{BFact1}. 

It is worth pointing out that by directly replacing~\eqref{solB} in~\eqref{BFact1}, we get the Palatini action,
\begin{equation}\label{Palatini}
	S[A,e]=\int_M d^n x \, e  \,  e^{\mu}{}_I e^{\nu}{}_J F_{\mu\nu}{}^{IJ},
\end{equation}
which provides an alternative way of establishing that~\eqref{BFact1} indeed describes general relativity.

\subsection{$[(n-1)+1]$-decomposition of the action}\label{decomp}

Splitting the objects involved in the action~\eqref{BFact1} into their temporal and spatial components, we obtain
\begin{eqnarray}
	&&\hspace{-8mm}S=\int_{\mathbb{I}} dt\int_{\Sigma}d^{n-1}x\Bigl[ \tilde{\Pi}^{aIJ}F_{taIJ}+\tilde{B}^{abIJ}F_{abIJ}+\Phi^{[d]IJKL}\left(2\uac{\eta}_{tabc[d]}\tilde{\Pi}^a{}_{IJ}\tilde{B}^{bc}{}_{KL}\right.\notag\\
	&&\hspace{-8mm}\left.+\tilde{\mu}_{[d]}{}^{[M]}\epsilon_{[M]IJKL}\right)+(n-4)\Phi^{t\langle e\rangle IJKL}\left(\uac{\eta}_{t\langle e\rangle abcd}\tilde{B}^{ab}{}_{IJ}\tilde{B}^{cd}{}_{KL}+\tilde{\mu}_{t\langle e\rangle}{}^{[M]}\epsilon_{[M]IJKL}\right)\Bigr],\label{BFact2}
\end{eqnarray}
where $\tilde{\Pi}^{aIJ}:=2\tilde{B}^{taIJ}$. Notice that the last two rows of~\eqref{BFact2} are only present for $n\geq 5$ since the whole term vanishes for $n=4$. The equations of motion for $\Phi^{[d]IJKL}$ and $\Phi^{t\langle e\rangle IJKL}$ are given by
\begin{subequations}
	\begin{eqnarray}
		\uac{\eta}_{tabc[d]}(\tilde{\Pi}^a{}_{IJ}\tilde{B}^{bc}{}_{KL}+\tilde{\Pi}^a{}_{KL}\tilde{B}^{bc}{}_{IJ})+\tilde{\mu}_{[d]}{}^{[M]}\epsilon_{[M]IJKL}&=&0,\label{constr1}\\
		\uac{\eta}_{t\langle e\rangle abcd}\tilde{B}^{ab}{}_{IJ}\tilde{B}^{cd}{}_{KL}+\tilde{\mu}_{t\langle e\rangle}{}^{[M]}\epsilon_{[M]IJKL}&=&0; \label{constr2}
	\end{eqnarray}
\end{subequations}
respectively. Our task is to solve these equations for $\tilde{\Pi}^a{}_{IJ}$, $\tilde{B}^{ab}{}_{IJ}$, $\tilde{\mu}_{[d]}{}^{[M]}$, and $\tilde{\mu}_{t\langle e\rangle}{}^{[M]}$, although the expressions of the last two are fixed once the expressions for the first two are known. Notice that~\eqref{constr1} is linear in 
$\tilde{\Pi}^a{}_{IJ}$ and, independently, also linear in $\tilde{B}^{ab}{}_{IJ}$; whereas~\eqref{constr2} only depends on $\tilde{B}^{ab}{}_{IJ}$ in a quadratic fashion. As mentioned above, equation~\eqref{constr2} only exists for $n\geq 5$, so, for $n=4$ we just have to deal with~\eqref{constr1}. We will see that, in order to solve both equations~\eqref{constr1} and~\eqref{constr2}, it is enough to consider just~\eqref{constr1}, which means that the set of equations~\eqref{constr2} gives rise to reducibility conditions for the whole system of equations (for $n>4$, of course).

Multiplying~\eqref{constr1} and~\eqref{constr2} by $\epsilon^{[N]IJKL}$, we obtain
\begin{subequations}
	\begin{eqnarray}
		\tilde{\mu}_{[d]}{}^{[M]}&=&-\frac{2\sigma}{4!(n-4)!}\uac{\eta}_{tabc[d]}\epsilon^{IJKL[M]}\tilde{\Pi}^a{}_{IJ}\tilde{B}^{bc}{}_{KL},\label{mu1}\\
		\tilde{\mu}_{t\langle e\rangle}{}^{[M]}&=&-\frac{\sigma}{4!(n-4)!}\uac{\eta}_{tabcd\langle e\rangle}\epsilon^{IJKL[M]}\tilde{B}^{ab}{}_{IJ}\tilde{B}^{cd}{}_{KL},\label{mu2}
	\end{eqnarray}
\end{subequations}
which express both $\tilde{\mu}_{[d]}{}^{[M]}$ and $\tilde{\mu}_{t\langle e\rangle}{}^{[M]}$ in terms of $\tilde{\Pi}^a{}_{IJ}$ and $\tilde{B}^{ab}{}_{IJ}$. These expressions are then substituted back into~\eqref{constr1} and~\eqref{constr2}. Note that the resulting expression from~\eqref{constr2} does not involve $\tilde{\Pi}^a{}_{IJ}$.

\subsection{Solution of the constraints}

In what follows we obtain the solution for $\tilde{\Pi}^a{}_{IJ}$ and $\tilde{B}^{ab}{}_{IJ}$ using only~\eqref{constr1}.~This is a remarkable fact because~\eqref{constr1} and~\eqref{constr2} define a coupled system for these variables. See the Appendix \ref{appendix} for a different approach where the solution of~\eqref{constr1} and~\eqref{constr2} is obtained from solving  first~\eqref{constr2} and then~\eqref{constr1}. So, let us start with~\eqref{constr1}. It can be equally rewritten as
\begin{equation}
	\tilde{\Pi}^{[a}{}_{IJ}\tilde{B}^{bc]}{}_{KL}+\tilde{\Pi}^{[a}{}_{KL}\tilde{B}^{bc]}{}_{IJ}+\tilde{\tilde{\mathcal{V}}}{}^{abc}{}_{IJKL}=0,\label{constr3}
\end{equation}
where we have defined
\begin{equation}
	\tilde{\tilde{\mathcal{V}}}{}^{abc}{}_{IJKL}:=\frac{1}{3!(n-4)!}\tilde{\eta}^{tabc[d]}\tilde{\mu}_{[d]}{}^{[M]}\epsilon_{[M]IJKL}.\label{V}
\end{equation}
As a consequence of this definition, $\tilde{\tilde{\mathcal{V}}}{}^{abc}{}_{IJKL}$ is totally antisymmetric in both spacetime and internal indices separately.

Equation~\eqref{constr3} is equivalent to the following equations:
\begin{subequations}
	\begin{eqnarray}
		&&\hspace{-10mm} K=I,\ L=J,\ I\neq J:\quad\tilde{\Pi}^{[a}{}_{IJ}\tilde{B}^{bc]}{}_{IJ}=0,\label{eq1}\\
		&&\hspace{-10mm} K=I,\  I\neq J\neq L:\quad \tilde{\Pi}^{[a}{}_{IJ}\tilde{B}^{bc]}{}_{IL}+\tilde{\Pi}^{[a}{}_{IL}\tilde{B}^{bc]}{}_{IJ}=0,\label{eq2}\\
		&&\hspace{-10mm} I \neq J \neq K \neq L:\quad \tilde{\Pi}^{[a}{}_{IJ}\tilde{B}^{bc]}{}_{KL}+\tilde{\Pi}^{[a}{}_{KL}\tilde{B}^{bc]}{}_{IJ}=\tilde{\Pi}^{[a}{}_{IK}\tilde{B}^{bc]}{}_{LJ}+\tilde{\Pi}^{[a}{}_{LJ}\tilde{B}^{bc]}{}_{IK}.\label{eq3}
	\end{eqnarray}
\end{subequations}
Notice that~\eqref{eq3} is obtained after performing a cyclic permutation of $(JKL)$ in~\eqref{constr3}, which leaves the term $\tilde{\tilde{\mathcal{V}}}{}^{abc}{}_{IJKL}$ invariant. Note also that there is no sum convention in the internal indices neither in~\eqref{eq1} nor in~\eqref{eq2}.

The solution of~\eqref{eq1},~\eqref{eq2}, and~\eqref{eq3} requires the use of the following two lemmas. The reader must be aware that in such lemmas the antisymmetry of $\tilde{\Pi}^a{}_{IJ}$ and $\tilde{B}^{ab}{}_{IJ}$--displayed in SubSect.~\eqref{decomp}--in the internal indices {\it is not} used. The antisymmetry property is imposed after lemma 1 and lemma 2 are established.

\begin{lemma}\label{lemma1}
	The vector density $\tilde{\Pi}^{a}{}_{IJ}$ and the bivector density $\tilde{B}^{ab}{}_{IJ}$ satisfy~\eqref{eq1} iff $\tilde{B}^{ab}{}_{IJ}=2\tilde{\Pi}^{[a}{}_{IJ} v^{b]}{}_{IJ}$, where $v^a{}_{IJ}$ is an arbitrary vector for any independent pair of indices $IJ$.
	
\end{lemma}
\begin{proof}
	That $\tilde{B}^{ab}{}_{IJ} = 2\tilde{\Pi}^{[a}{}_{IJ} v^{b]}{}_{IJ}$ is a sufficient condition for~\eqref{eq1} follows immediately from the repetition of the factor $\tilde{\Pi}^{a}{}_{IJ}$ when this condition is substituted into~\eqref{eq1}. To see that  $\tilde{B}^{ab}{}_{IJ} = 2\tilde{\Pi}^{[a}{}_{IJ} v^{b]}{}_{IJ}$ is also a necessary condition for~\eqref{eq1}, notice that we can rewrite~\eqref{eq1} as
	\begin{equation}
		\tilde{\Pi}^{a}{}_{IJ} \tilde{B}^{bc}{}_{IJ} + \tilde{\Pi}^{b}{}_{IJ} \tilde{B}^{ca}{}_{IJ} + \tilde{\Pi}^{c}{}_{IJ} \tilde{B}^{ab}{}_{IJ} =0.
	\end{equation}
	Multiplying this expression by a nonvanishing 1-form $\chi_c$ such that $\tilde{\Pi}^{c}{}_{IJ} \chi_c\neq 0$, we obtain
	\begin{equation}\label{BPiv}
		\tilde{B}^{ab}{}_{IJ} = 2 \tilde{\Pi}^{[a}{}_{IJ} v^{b]}{}_{IJ},
	\end{equation}
	where we have defined $v^a{}_{IJ}:=- ( \tilde{B}{}^{ab}{}_{IJ} \chi_b ) / (\tilde{\Pi}^{c}{}_{IJ} \chi_c) $.
\end{proof}

Once we have solved~\eqref{eq1}, we now focus on~\eqref{eq2}. Because $I$ is fixed and $I\neq J\neq L$ in~\eqref{eq2}, such a condition indicates how two pairs $(\tilde{\Pi}_1^{a},\tilde{B}_1^{ab})$ and $(\tilde{\Pi}_2^{a},\tilde{B}_2^{ab})$, where we have identified $IJ$ with 1 and $IL$ with 2, are related to one another. In fact, assuming they both independently satisfy~\eqref{eq1}, we have, because of Lemma~\ref{lemma1}, that $\tilde{B}^{ab}_i=2\tilde{\Pi}^{[a}_iv^{b]}_i$ ($i=1,2$), where $v_1\neq v_2$ in general. Rewriting~\eqref{eq2} as
\begin{equation}\label{Pi1Pi2}
	\tilde{\Pi}^{[a}_1\tilde{B}^{bc]}_2+\tilde{\Pi}^{[a}_2\tilde{B}^{bc]}_1=0,
\end{equation}
we have the following result:
\begin{lemma}\label{lemma2}
	Suppose that $(\tilde{\Pi}_i^{a},\tilde{B}_i^{ab})$ ($i=1,2$) satisfy Lemma~\ref{lemma1}. $(\tilde{\Pi}_i^{a},\tilde{B}_i^{ab})$ ($i=1,2$) fulfill~\eqref{Pi1Pi2} iff  $\tilde{B}^{ab}_i=2\tilde{\Pi}^{[a}_i w^{b]}$ ($i=1,2$) for some arbitrary vector $w$. Thus, $\tilde{B}_1$ and $\tilde{B}_2$ share the same vector $w$.
\end{lemma}
\begin{proof}
	That $\tilde{B}^{ab}_i=2\tilde{\Pi}^{[a}_i w^{b]}$ ($i=1,2$) is a sufficient condition for~\eqref{Pi1Pi2} follows from their direct subtitution in~\eqref{Pi1Pi2}. On the other side, as mentioned above, since $(\tilde{\Pi}_i^{a},\tilde{B}_i^{ab})$ ($i=1,2$) fulfill Lemma~\eqref{lemma1}, we have $\tilde{B}^{ab}_i=2\tilde{\Pi}^{[a}_iv^{b]}_i$ ($i=1,2$). Replacing them into~\eqref{Pi1Pi2}, we obtain
	\begin{equation}
		\tilde{\Pi}^{[a}_1\tilde{\Pi}^{b}_2(v^{c]}_2-v^{c]}_1)=0.
	\end{equation}
	Since $\tilde{\Pi}_1\neq \tilde{\Pi}_2$ in general, then last equation implies that $v^{c}_2=v^{c}_1+ \uac{\alpha} \tilde{\Pi}^{c}_1+\uac{\beta} \tilde{\Pi}^{c}_2$, where $\uac{\alpha}$ and $ \uac{\beta}$ are arbitrary real-valued functions of weight $-1$. Substituting this expression into $\tilde{B}_2$ we obtain
	\begin{equation}
		\tilde{B}^{ab}_2=2\tilde{\Pi}^{[a}_2(v^{b]}_1+\uac{\alpha}\tilde{\Pi}^{b]}_1).
	\end{equation}
	Defining $w:=v_1+\uac{\alpha}\tilde{\Pi}_1$, we have $\tilde{B}^{ab}_1=2\tilde{\Pi}^{[a}_1 w^{b]}$ and $\tilde{B}^{ab}_2=2\tilde{\Pi}^{[a}_2 w^{b]}$, i.e.,  $\tilde{B}^{ab}_i=2\tilde{\Pi}^{[a}_i w^{b]}$ ($i=1,2$).
\end{proof}

For fixed $I$, let us consider a third element $(\tilde{\Pi}_3^{a},\tilde{B}_3^{ab})$ that satisfies both~\eqref{eq1} and~\eqref{eq2}. Because of Lemma~\ref{lemma1}, we have that $\tilde{B}^{ab}_3=2\tilde{\Pi}^{[a}_3 v_3^{b]}$ for some vector $v_3$. Now, because of Lemma~\ref{lemma2}, $\tilde{B}^{ab}_3$ shares a common vector with both $\tilde{B}^{ab}_1$ and $\tilde{B}^{ab}_2$ (which are not necessarily equal). Since $\tilde{\Pi}_i$ ($i=1,2,3$) are assumed to be independent, this means that we must have $v_3=w$. Therefore, $\tilde{B}^{ab}_i=2\tilde{\Pi}^{[a}_i w^{b]}$ ($i=1,2,3$). This procedure can be continued until we have covered all the possible values of $J$ and $L$ for fixed $I$ in~\eqref{eq2}, so that $\tilde{B}^{ab}_i=2\tilde{\Pi}^{[a}_i w^{b]}$ ($i=1,2,\ldots$). For instance, for $I=0$ and $J\neq0$ the previous result implies that $\tilde{B}^{ab}{}_{0J}=2\tilde{\Pi}^{[a}{}_{0J} w^{b]}{}_0$ for some common vector $w_0$; for $I=1$ and $J\neq1$ we have $\tilde{B}^{ab}{}_{1J}=2\tilde{\Pi}^{[a}{}_{1J} w^{b]}{}_1$ for some common vector $w_1$; and so on. In consequence, equations~\eqref{eq1} and~\eqref{eq2} allow us to conclude that
\begin{equation}\label{BIJ}
	\tilde{B}^{ab}{}_{IJ}=2\tilde{\Pi}^{[a}{}_{IJ}w^{b]}{}_I,
\end{equation}
where $w_I$ is the common vector shared by all the $\tilde{B}^{ab}{}_{IJ}$ with fixed $I$ and $J\neq I$.

We now impose antisymmetry in the internal indices. Equation~\eqref{BIJ} implies
\begin{equation}\label{BIJ2}
	\tilde{B}^{ab}{}_{JI}=2\tilde{\Pi}^{[a}{}_{JI}w^{b]}{}_J=-2\tilde{\Pi}^{[a}{}_{IJ}w^{b]}{}_J.
\end{equation}
Since $\tilde{B}^{ab}{}_{JI}=-\tilde{B}^{ab}{}_{IJ}$,~\eqref{BIJ} and~\eqref{BIJ2} imply, for fixed $I$ and $J$, the relation
\begin{equation}
	\tilde{\Pi}^{[a}{}_{IJ}(w^{b]}{}_I-w^{b]}{}_J)=0,
\end{equation}
which means that $\tilde{\Pi}^{a}{}_{IJ}$, as the components of a spatial vector, are proportional to the difference $w^a{}_I-w^a{}_J$. Therefore, we must have
\begin{equation}\label{PiIJ}
	\tilde{\Pi}^{a}{}_{IJ}=\tilde{\theta}_{IJ}(w^a{}_I-w^a{}_J),
\end{equation}
for some real-valued functions $\tilde{\theta}_{IJ}$ of weight $+1$. Given that $\tilde{\Pi}^{a}{}_{IJ}=-\tilde{\Pi}^{a}{}_{JI}$, $\tilde{\theta}_{IJ}$ is symmetric in $IJ$, $\tilde{\theta}_{IJ}=\tilde{\theta}_{JI}$.

Substituting~\eqref{PiIJ} in~\eqref{BIJ} yields
\begin{equation}\label{BIJ3}
	\tilde{B}^{ab}{}_{IJ}=2\tilde{\theta}_{IJ}w^{[a}{}_Iw^{b]}{}_J.
\end{equation}
To fix $\tilde{\theta}_{IJ}$ we use the remaining equation~\eqref{eq3}. Replacing~\eqref{PiIJ} and~\eqref{BIJ3} in~\eqref{eq3} we get
\begin{eqnarray}
	&&\left(\tilde{\theta}_{IJ}\tilde{\theta}_{KL}-\tilde{\theta}_{IK}\tilde{\theta}_{JL}\right)\left(w^{[a}{}_Iw^{b}{}_Kw^{c]}{}_L-w^{[a}{}_Jw^{b}{}_Kw^{c]}{}_L+w^{[a}{}_Kw^{b}{}_Iw^{c]}{}_J-w^{[a}{}_Lw^{b}{}_Iw^{c]}{}_J\right)=0.\notag\\
\end{eqnarray}
Since the second factor is nonvanishing in general (otherwise we would have relations among the different $w_I$'s), we conclude that
\begin{equation}\label{theta}
	\tilde{\theta}_{IJ}\tilde{\theta}_{KL}=\tilde{\theta}_{IK}\tilde{\theta}_{JL}.
\end{equation}
Although $I\neq J\neq K\neq L$ in this equation, we may assume that it holds even if some of the indices are repeated because those cases will not contribute to the final expressions for $\tilde{\Pi}^{a}{}_{IJ}$ and $\tilde{B}^{ab}{}_{IJ}$.

Since $\tilde{\theta}_{IJ}$ is a symmetric matrix, it can be (orthogonally) diagonalized, which means that it has at least one nonvanishing eigenvector $\lambda^I$ associated to some given nonvanishing eigenvalue whose value is not important for our purpose. This means that there exists $\lambda^I$ such that $\tilde{\theta}_{IJ}\lambda^I\lambda^J\neq 0$. Multiplying~\eqref{theta} by $\lambda^K\lambda^L$ we obtain
\begin{equation}\label{theta1}
	\tilde{\theta}_{IJ}=\epsilon M_IM_J,
\end{equation}
where $\epsilon:=\text{sgn}(\tilde{\theta}_{IJ}\lambda^I\lambda^J)$ and $M_I:=\tilde{\theta}_{IJ}\lambda^J/|\tilde{\theta}_{KL}\lambda^K\lambda^L|^{1/2}$ has weight $+1/2$. Inserting~\eqref{theta1} into~\eqref{PiIJ} and ~\eqref{BIJ3}, we arrive at
\begin{subequations}
	\begin{eqnarray}
		\tilde{\Pi}^{a}{}_{IJ}&=&2\epsilon V^a{}_{[I}M_{J]},\label{SolPi}\\
		 \tilde{B}^{ab}{}_{IJ}&=&2\epsilon V^{[a}{}_{I}V^{b]}{}_{J},\label{BIJ4}
	\end{eqnarray}
\end{subequations}
where $V^a{}_{I}:=M_Iw^a{}_{I}$ has weight $+1/2$, i.e., we have rescaled $w^a{}_{I}$. Notice that in order to obtain this solution we did not use~\eqref{constr2}, which means that these constraints give rise to reducibility conditions for the whole system of equations. Actually,~\eqref{BIJ4} corresponds to the solution of~\eqref{constr2}, so that it is automatically satisfied [for the value of $\tilde{\mu}_{t\langle e\rangle}{}^{[M]}$ given by~\eqref{mu2}]. We point out that the sign $\epsilon$ appearing on the right-hand side of~\eqref{SolPi} and~\eqref{BIJ4} reflects the quadratic nature of the constraints~\eqref{eqmot3}. In~\eqref{SolPi}, this sign can be eliminated by absorbing it in either $V^a{}_I$ or $M_I$ (after all, they are arbitrary); on the other hand, since~\eqref{BIJ4} is quadratic in $V^a{}_I$, that sign cannot be eliminated from it.

Defining the norm of $M^I$ as $|M|:=|M_I M^I|^{1/2}$ (of weight $+1/2$), let us introduce the unit vector $m^I$ along $M^I$:
\begin{equation}\label{mI}
	m^I:=\frac{1}{|M|}M^I,
\end{equation}
which satisfies $m_Im^I=\sigma$ for $\sigma:
=\text{sgn} (M_I M^I)$. Although in principle they do not need to be related, we choose this value to be equal to the metric component $\eta_{00}$ for consistency. Multiplying~\eqref{SolPi} by $m^J$ and defining $\tilde{\Pi}^a{}_{I}:=-\sigma\tilde{\Pi}^{a}{}_{IJ}m^J$, we find
\begin{equation}\label{PiV}
	\tilde{\Pi}^a{}_{I}=\epsilon|M|\left(-V^a{}_I+\sigma V^a{}_Jm^Jm_I\right).
\end{equation}
Note that the orthogonality condition $\tilde{\Pi}^a{}_{I}m^I=0$ is fulfilled. In addition, the sign $\epsilon$ on the right of~\eqref{SolPi} is absorbed into the definition of $\tilde{\Pi}^a{}_{I}$. We can rewrite~\eqref{PiV} as
\begin{equation}\label{VaI}
	V^a{}_I= -\epsilon |M|^{-1} \tilde{\Pi}^a{}_{I}+\sigma V^a{}_Jm^Jm_I,
\end{equation}
which provides a decomposition of $V^a{}_I$ along $\tilde{\Pi}^a{}_{I}$ and $m^I$. Replacing~\eqref{VaI} in~\eqref{SolPi} gives
\begin{equation}\label{SolPi2}
	\tilde{\Pi}^{a}{}_{IJ}=-2\tilde{\Pi}^a{}_{[I}m_{J]}.
\end{equation}
On the other hand, inserting~\eqref{VaI} in~\eqref{BIJ4} leads to
\begin{eqnarray}
	\tilde{B}^{ab}{}_{IJ} &=& 2\epsilon |M|^{-2}\tilde{\Pi}^{[a}{}_{I}\tilde{\Pi}^{b]}{}_{J} +2 \sigma|M|^{-1} m^K V^{[a}{}_K(\tilde{\Pi}^{b]}{}_{I}m_{J}-\tilde{\Pi}^{b]}{}_{J}m_{I}).\label{BIJ5}
\end{eqnarray}
Now define $\uac{N}:=2\epsilon\sigma |M|^{-2}$ and $N^a:=2\sigma|M|^{-1} V^{a}{}_Im^I$. With this,~\eqref{BIJ5} reads
\begin{equation}\label{BIJ6}
	\tilde{B}^{ab}{}_{IJ}=N^{[a}(\tilde{\Pi}^{b]}{}_{I}m_{J}-\tilde{\Pi}^{b]}{}_{J}m_{I})+\sigma \uac{N} \tilde{\Pi}^{[a}{}_{I}\tilde{\Pi}^{b]}{}_{J}.
\end{equation}
Equations~\eqref{SolPi2} and~\eqref{BIJ6} are the desired solutions for $\tilde{\Pi}^{a}{}_{IJ}$ and $\tilde{B}^{ab}{}_{IJ}$ in terms of the $n^2$ objects $\tilde{\Pi}^{a}{}_{I}$, $\uac{N}$, and $N^a$. It is worth pointing out that $m^I$ is not an independent variable since the orthogonaly condition and its unit character imply that it can be expressed entirely in terms of $\tilde{\Pi}^{a}{}_{I}$ (up to a sign) as
\begin{equation}
	\label{mPi}
	m_{I}=\frac{1}{(n-1)! \sqrt{h}}\epsilon_{IJ_1 \ldots J_{n-1}}\underaccent{\tilde}{\eta}_{t a_1 \ldots a_{n-1}} \tilde{\Pi}^{a_1 J_1}  \cdots \tilde{\Pi}^{a_{n-1} J_{n-1}},
\end{equation}
where $h:=\mathrm{det}(\tilde{\tilde{h}}{}^{ab})$, for $\tilde{\tilde{h}}{}^{ab}:=\tilde{\Pi}^{a I}\tilde{\Pi}^{b}{}_{I}$, is positive definite and has weight $2(n-2)$. This latter object allows us to bring in the spatial metric on $\Sigma$, which can be defined as $q_{ab}:=h^{\frac{1}{n-2}}\uac{\uac{h}}{}_{ab}$ with $\uac{\uac{h}}{}_{ab}$ the inverse of $\tilde{\tilde{h}}{}^{ab}$~\cite{PalatininD}.  

To close this section, let us remark again that the equations~\eqref{constr2} were not involved to find the desired solutions~\eqref{SolPi2} and~\eqref{BIJ6}. This means that their only role is just to fix the multiplier $\tilde{\mu}_{t\langle e\rangle}{}^{[M]}$. Althought they are not necessary, we give for the sake of completeness the ensuing solutions for the Lagrange multipliers $\tilde{\mu}_{[d]}{}^{[M]}$ and $\tilde{\mu}_{t\langle e\rangle}{}^{[M]}$. Replacing~\eqref{SolPi2} and~\eqref{BIJ6} in~\eqref{mu1} and~\eqref{mu2}, we get, after making use of~\eqref{mPi},
\begin{subequations}
	\begin{eqnarray}
		\hspace{-8mm}\tilde{\mu}_{[d]}{}^{[M]}&=&-\sigma \sqrt{h}\uac{N}\uac{\uac{h}}{}_{[d_1|c_1}\uac{\uac{h}}{}_{|d_2|c_2}\cdots\uac{\uac{h}}{}_{|d_{n-4}]c_{n-4}}\tilde{\Pi}^{c_1M_1}\tilde{\Pi}^{c_2M_2} \cdots \tilde{\Pi}^{c_{n-4}M_{n-4}},\label{mu3}\\
		\hspace{-8mm}\tilde{\mu}_{t\langle e\rangle}{}^{[M]}&=&N^a\tilde{\mu}_{a\langle e\rangle}{}^{[M]}-\sigma \sqrt{h}\uac{N}^2\uac{\uac{h}}{}_{e_1a_1} \cdots\uac{\uac{h}}{}_{e_{n-5}a_{n-5}} m^{[M_1|}\tilde{\Pi}^{a_1|M_2|} \cdots \tilde{\Pi}^{a_{n-5}|M_{n-4}]}.\label{mu4}
	\end{eqnarray}
\end{subequations}
Note that the sets of totally antisymmetric indices have been explicitly displayed on the right-hand side of~\eqref{mu3} and~\eqref{mu4} because such indices do not belong to the same object. This notation will we used from now on throughout the paper.

As expressed above, the solution of~\eqref{constr1} and~\eqref{constr2} involves $n^2$ free objects $\tilde{\Pi}^{a}{}_{I}$, $\uac{N}$, and $N^a$ [or $V^a{}_I$ and $M_I$ according to~\eqref{SolPi} and~\eqref{BIJ4}]. However, it can be verified that for $n>4$ the number of equations in~\eqref{constr1} and~\eqref{constr2} surpasses the number of independent unknowns involved (they are exactly equal in $n=4$). This means that not all the equations~\eqref{constr1} and~\eqref{constr2} are independent (so that the system is reducible), and indeed we found that~\eqref{constr2} was not necessary to find the solution for $\tilde{\Pi}^{aIJ}$ and $\tilde{B}^{abIJ}$: once ~\eqref{constr1} was solved, the only purpose of~\eqref{constr2} was to fix $\tilde{\mu}_{t\langle e\rangle}{}^{[M]}$. To quantify the number of reducibility relations between~\eqref{constr1} and~\eqref{constr2}, we know that the number of free variables ($FV$) at the end must be equal to the difference between the number of independent unknowns ($IU$, including $\tilde{\Pi}^{aIJ}$, $\tilde{B}^{abIJ}$, $\tilde{\mu}_{[d]}{}^{[M]}$, and $\tilde{\mu}_{t\langle e\rangle}{}^{[M]}$) and the number of independent equations. In turn, the number of independent equations is equal to the number of initial equations ($IE$)~\eqref{constr1} and~\eqref{constr2} minus the number of reducibility relations ($RR$) among them. In consequence, we obtain that $RR=FV+IE-IU$. Since $FV=n^2$, $IE=\frac{1}{2}{n\choose 4}{n\choose 2}\left[{n\choose 2}+1\right]$ and $IU={n\choose 2}^2+{n\choose 4}^2$, we obtain $RR=\frac{1}{288}(n-4)(n-3)n^2(n+1)(n+2)(n^2-2n+9)$. Hence, the number of independent equation is given by $\frac{1}{576}(n-3)n^2(n^5-9n^4+31n^3-51n^2+184n+132)$. Notice that the number of reducibility relations vanishes for $n=4$, as expected. The fact that there is no need to know the explicit form of the reducibility relations to obtain the above solutions for the involved variables in the general case is quite remarkable. More details on the issue of reducibility in the Lagrangian framework can be found in Ref.~\cite{Freid_Puz}.

\subsection{Back to the action}

Substituting~\eqref{SolPi2} and~\eqref{BIJ6} in the action~\eqref{BFact2}, we obtain, after integrating by parts the first term and neglecting boundary terms,
\begin{eqnarray}\label{action_Pi}
	S &=&\int_{\mathbb{I}} dt\int_{\Sigma}d^{n-1}x \left ( -2\tilde{\Pi}^{a I}m^{J}\partial_{t}A_{a I J} + A_{t I J}\tilde{\mathcal{G}}^{I J} - N^{a}\tilde{\mathcal{V}}_{a} - \underaccent{\tilde}{N} \tilde{\tilde{\mathcal{C}}} \right ),
\end{eqnarray}
with 
\begin{subequations}
	\begin{eqnarray}
		\label{gauss} \tilde{\mathcal{G}}^{IJ}&:=&-2D_a\bigl(\tilde{\Pi}^{a[I}m^{J]}\bigr),\\
		\label{vector} \tilde{\mathcal{V}}_{a}&:=& -2\tilde{\Pi}^{bI}m^{J} F_{a b I J}, \\
		\label{scalar} \tilde{\tilde{\mathcal{C}}} &:=& -\sigma \tilde{\Pi}^{a I}\tilde{\Pi}^{b J}F_{a b I J}.
	\end{eqnarray}
\end{subequations}
The action~\eqref{action_Pi} coincides, up to a constant global factor, with the expression (11) of Ref.~\cite{PalatininD} for a vanishing cosmological constant, which corresponds to an intermediate step in the canonical analysis of the $n$-dimensional Palatini action. From this point on and in order to avoid the introduction of second-class constraints, the canonical analysis continues as in Ref.~\cite{PalatininD}. Although all the details can be found there, let us briefly summarize the ongoing procedure. Since the amount of variables in the connection $A_{aIJ}$ is greater than the number of variables in $\tilde{\Pi}^{aI}$, the theory is endowed with a presymplectic structure whose null directions must be isolated in order to have a well-defined symplectic structure. This is achieved by reparametrizing the connection in terms of an equivalent set of variables $(\mathcal{Q}_{aI},\uac{\uac{\lambda}}_{abc})$ as $A_{aIJ} = M_{a}{}^{b}{}_{IJK} \mathcal{Q}_{b}{}^{K} +  \tilde{\tilde{N}}_a{}^{bcd}{}_{IJ}\uac{\uac{\lambda}}_{bcd}$, where the objects $M_{a}{}^{b}{}_{IJK}$ and $\tilde{\tilde{N}}_a{}^{bcd}{}_{IJ}$ are constructed solely from $\tilde{\Pi}^{aI}$. It turns out that the variables $({\cal{Q}}_{aI},\tilde{\Pi}^{aI})$ label the points of the kinematic phase space of the theory, whereas the variables $\uac{\uac{\lambda}}_{abc}$ play the role of auxiliary variables that can be eliminated from the action via their own dynamics. The final result is that the action takes the form
\begin{eqnarray}\label{final_action}
	S &=&\int_{\mathbb{I}} dt\int_{\Sigma}d^{n-1}x\left ( 2\tilde{\Pi}^{a I}\partial_{t} {\cal{Q}}_{aI} -\Lambda_{IJ} \tilde{\mathcal{G}}^{I J} -2N^{a}\tilde{\mathcal{D}}_{a} - \underaccent{\tilde}{N} \tilde{\tilde{\mathcal{H}}} \right ),
\end{eqnarray}
where $\tilde{\mathcal{G}}^{I J}$, $\tilde{\mathcal{D}}_{a}$, and $\tilde{\tilde{\mathcal{H}}}$ are, correspondingly, the Gauss, diffeomorphism, and scalar constraints. These are the only constraints of the theory which, being first class, generate the two underlying gauge symmetries of first-order gravity: local $SO(n-1,1)$ or $SO(n)$ transformations and spacetime diffeomorphisms. In terms of the canonical variables, they read
\begin{subequations}
	\begin{eqnarray}
		\label{Gauss2} \tilde{\mathcal{G}}^{IJ} &=&  2 \tilde{\Pi}^{a [I} \mathcal{Q}_{a}{}^{J]} + 4 \delta^{I}_{[K}\delta^{J}_{L]} \tilde{\Pi}^{a [K} n^{M]} \Gamma_{a}{}^{L}{}_{M}, \\
		\label{diffeomorphism} \tilde{\mathcal{D}}_{a} &=& 2\tilde{\Pi}^{b I} \partial_{[a} \mathcal{Q}_{b] I} - \mathcal{Q}_{a}{}^{I}\partial_{b}\tilde{\Pi}^{b}{}_{I}, \\ 
		\tilde{\tilde{\mathcal{H}}} &:=& - \sigma \tilde{\Pi}^{a I}\tilde{\Pi}^{b J}R_{a b I J} +  2 \tilde{\Pi}^{a [I}\tilde{\Pi}^{|b|J]} \left ( \mathcal{Q}_{aI} \mathcal{Q}_{bJ} +2 \mathcal{Q}_{aI} \Gamma_{bJK} n^{K} + \Gamma_{aIK} \Gamma_{bJL} n^{K} n^{L} \right ),\nonumber\\
		\label{scalar2}
	\end{eqnarray}
\end{subequations}
with $\Gamma_{aIJ}$ being the $SO(n-1,1)$ or $SO(n)$ connection compatible with $\tilde{\Pi}^{aI}$, i.e., $\nabla_{a}\tilde{\Pi}^{b I}= \partial_{a}\tilde{\Pi}^{b I} + \Gamma_{a}{}^{I}{}_{J}\tilde{\Pi}^{b J} + \Gamma^{b}{}_{a c}\tilde{\Pi}^{c I} - \Gamma^{c}{}_{a c}\tilde{\Pi}^{b I} =0$, and $R_{abIJ}:= \partial_{a}\Gamma_{bIJ} - \partial_{b}\Gamma_{aIJ} + \Gamma_{aIK} \Gamma_b{}^K{}_J - \Gamma_{bIK} \Gamma_a{}^K{}_J$ being its curvature.

\section{Freidel-Krasnov-Puzio model}

\subsection{The action}

The original action of Ref.~\cite{Freid_Puz} has the same structure of~\eqref{BFact1}, but now the antisymetrization in the internal indices of the product of $B$'s is regarded. It is given by
	\begin{eqnarray}
		S[A,\tilde{B},\uac{\Phi},\tilde{\nu}]=&&\int_M d^nx\left(\tilde{B}^{\mu\nu IJ}F_{\mu\nu IJ}+\uac{\Phi}_{\mu\nu\lambda\sigma[M]}\epsilon^{[M]IJKL}\tilde{B}^{\mu\nu}{}_{IJ}\tilde{B}^{\lambda\rho}{}_{KL}\right.\notag\\
		&&\left.+\tilde{\nu}_{[\alpha]}{}^{[M]}\tilde{\eta}^{[\alpha]\mu\nu\lambda\sigma}\uac{\Phi}_{\mu\nu\lambda\sigma[M]}\right),\label{BFact3}
	\end{eqnarray}
where both $\uac{\Phi}_{\mu\nu\lambda\sigma[M]}$ (of weight $-1$) and $\tilde{\nu}_{[\alpha]}{}^{[M]}$ (of weight $1$) are Lagrange multipliers, the former having the following symmetries in the spacetime indices: $\uac{\Phi}_{\mu\nu\lambda\sigma[M]}=-\uac{\Phi}_{\nu\mu\lambda\sigma[M]}=-\uac{\Phi}_{\mu\nu\sigma\lambda[M]}=\uac{\Phi}_{\lambda\sigma\mu\nu[M]}$. Consequently, $\uac{\Phi}_{\mu\nu\lambda\sigma[M]}$ has the same symmetries in the spacetime indices as $\Phi_{[\alpha]IJKL}$ of Sec.~\ref{secBF1} does in the internal indices. It can be shown, in a completely analogous fashion as we did for the action~\eqref{BFact1}, that the equations of motion arising from~\eqref{BFact3} also lead to the Einstein equations without a cosmological constant.

\subsection{$[(n-1)+1]$-decomposition}

Splitting space and time in~\eqref{BFact3} yields

\begin{eqnarray}
	S=&&\int_{\mathbb{I}} dt\int_{\Sigma}d^{n-1}x\Bigl[ \tilde{\Pi}^{aIJ}F_{taIJ}+\tilde{B}^{abIJ}F_{abIJ}+\uac{\Phi}_{ta tb[M]}\epsilon^{[M]IJKL}\tilde{\Pi}^a{}_{IJ}\tilde{\Pi}^b{}_{KL}\notag\\
	&&+2\uac{\Phi}_{tabc[M]}\left(\epsilon^{[M]IJKL}\tilde{\Pi}^a{}_{IJ}\tilde{B}^{bc}{}_{KL}+2\tilde{\nu}_d{}^{[M]}\tilde{\eta}^{tabc[d]}\right)\notag\\
	&&+\uac{\Phi}_{abcd[M]}\left(\epsilon^{[M]IJKL}\tilde{B}^{ab}{}_{IJ}\tilde{B}^{cd}{}_{KL}+(n-4)\tilde{\nu}_{t\langle e\rangle}{}^{[M]}\tilde{\eta}^{tabcd\langle e\rangle}\right)\Bigr],\label{BFact4}
\end{eqnarray}
where $\tilde{\Pi}^{aIJ}=2\tilde{B}^{taIJ}$ as above. Note that in contrast to~\eqref{BFact2}, the fact that the multiplier $\uac{\Phi}$ is not totally antisymmetric in the spacetime indices produces one more term in the action (the one involving $\uac{\Phi}_{ta tb[M]}$). The variation of~\eqref{BFact4} with respect to $\uac{\Phi}_{ta tb[M]}$, $\uac{\Phi}_{tabc[M]}$, and $\uac{\Phi}_{abcd[M]}$ gives, correspondingly, the constraints
\begin{subequations}
	\begin{eqnarray}
		\epsilon^{[M]IJKL}\tilde{\Pi}^a{}_{IJ}\tilde{\Pi}^b{}_{KL}&=&0,\label{constr12}\\
		\epsilon^{[M]IJKL}\tilde{\Pi}^a{}_{IJ}\tilde{B}^{bc}{}_{KL}+2\tilde{\nu}_{[d]}{}^{[M]}\tilde{\eta}^{tabc[d]}&=&0,\label{constr22}\\
		 \epsilon^{[M]IJKL}\tilde{B}^{ab}{}_{IJ}\tilde{B}^{cd}{}_{KL}+(n-4)\tilde{\nu}_{t\langle e\rangle}{}^{[M]}\tilde{\eta}^{tabcd\langle e\rangle}&=&0.\label{constr32}
	\end{eqnarray}
\end{subequations}
In this case, whereas~\eqref{constr12} is independent of $\tilde{B}^{abIJ}$ and~\eqref{constr32} is independent of $\tilde{\Pi}^{aIJ}$, the constraint~\eqref{constr22} blends both components and thereby fixes the relationship between them. Remarkably, the constraint~\eqref{constr12} does not involve the multiplier $\tilde{\nu}$, which simplifies things a little bit. 

\subsection{Solution of the constraints}

Although to solve~\eqref{constr12}-\eqref{constr32} for the involved variables we can rely on Ref.~\cite{Freid_Puz}, which would result in a process similar to that followed to solve~\eqref{constr1}, here we shall follow a different approach, a more algebraic one. Before going to the details, let us point out that it is equation~\eqref{constr12} that is promoted to a constraint in the canonical analysis of the action~\eqref{BFact3} carried out in Ref.~\cite{Thiemann_2013}, which later leads to the presence of second-class constraints in the theory. In contrast, here we shall solve~\eqref{constr12} from the very beginning, thus dodging the introduction of second-class constraints.

The solution of~\eqref{constr12} takes exactly the same form as~\eqref{SolPi2}~\cite{Bodendorfer_2013}, with $m^I$ given by~\eqref{mPi} (so that $m_I\tilde{\Pi}^{aI}=0$ and $m_Im^I=\sigma$ are satisfied). Now we turn our attention to the constraint~\eqref{constr22}. Multiplying it by $\epsilon_{[M]PQRS} \uac{\uac{h}}{}_{ae}\tilde{\Pi}^{eRS}$, we obtain, after a bit of algebra,
\begin{eqnarray}
	&&\tilde{B}^{bc}{}_{PQ}=2 \lambda^{bc}{}_d \tilde{\Pi}^d{}_{[P}m_{Q]}-\frac{1}{4(n-3)!}\tilde{\nu}_{[d]}{}^{[M]}\epsilon_{[M]PQRS}\uac{\uac{h}}{}_{ae}\tilde{\Pi}^{eRS}\tilde{\eta}^{tabc[d]},\label{B1}
\end{eqnarray}
where both~\eqref{SolPi2} and the identity $\uac{\uac{h}}{}_{ab}\tilde{\Pi}^{aI}\tilde{\Pi}^b{}_J=\delta^I_J-\sigma m^I m_J$
have been used, and we have defined $\lambda^{bc}{}_d:=-\sigma m^I\tilde{B}^{bc}{}_{IJ}\tilde{\Pi}^{aJ}\uac{\uac{h}}{}_{ad}$. To simplify the last term of~\eqref{B1}, notice that we can solve~\eqref{constr22} for $\tilde{\nu}_{[d]}{}^{[M]}$, giving
\begin{equation}\label{nu1}
	\tilde{\nu}_{[d]}{}^{[M]}=-\frac{2}{4!(n-4)!}\uac{\eta}_{tabc[d]}\epsilon^{IJKL[M]}\tilde{\Pi}^a{}_{IJ}\tilde{B}^{bc}{}_{KL}.
\end{equation}
The combination of this expression with~\eqref{SolPi2} implies $m_I\tilde{\nu}_{[d]}{}^{I M_1\ldots M_{n-5}}=0$, which holds for the contraction of $m_I$ with any of the internal indices of $\tilde{\nu}_{[d]}{}^{[M]}$ because they are totally antisymmetric.

Replacing~\eqref{B1} in~\eqref{nu1}, we get
\begin{eqnarray}
	\tilde{\nu}_{[d]}{}^{[M]}=&&\frac{(7-n)}{3(n-3)}\tilde{\nu}_{[d]}{}^{[M]}+\frac{(n-4)}{3(n-3)} \sum_{i=1}^{n-4}(-1)^{i-1}\uac{\uac{h}}{}_{f d_i}\tilde{\Pi}^{f[M_1|} \tilde{\Pi}^a{}_I\tilde{\nu}_{a d_1\ldots \widehat{d_i} \ldots d_{n-4}}{}^{I|M_2\ldots M_{n-4}]}.\notag\\
	\label{nu2}
\end{eqnarray}
where ``$\widehat{d_i}$'' indicates that the subscript $d_i$ is missing in the corresponding object. Note that $\lambda^{bc}{}_d$ is not present in~\eqref{nu2}, which allows us to build a recurrence relation for $\tilde{\nu}_{[d]}{}^{[M]}$ displayed below in~\eqref{nu3}.  Note also that for $n=4$ the last term in~\eqref{nu2} is not present and the equation reduces to the tautology ``$\tilde{\nu}=\tilde{\nu}$''. On the other hand, for $n>4$ the solution of~\eqref{nu2} for $\tilde{\nu}_{[d]}{}^{[M]}$ gives
\begin{eqnarray}
	\tilde{\nu}_{[d]}{}^{[M]}=&&\frac14 \sum_{i=1}^{n-4}(-1)^{i-1}\uac{\uac{h}}{}_{f d_i}\tilde{\Pi}^{f[M_1|} \tilde{\Pi}^a{}_I\tilde{\nu}_{a d_1\ldots \widehat{d_i} \ldots d_{n-4}}{}^{I|M_2 \ldots M_{n-4}]},\label{nu3}
\end{eqnarray}
which can alternatively be rewritten in a more compact form as
\begin{equation}
	\tilde{\nu}_{[d]}{}^{[M]}=\frac{(n-4)}{4}\tilde{\Pi}^{f[M_1|}\uac{\uac{h}}{}_{f[d_1|}\tilde{\Pi}^a{}_I\tilde{\nu}_{a | d_2 \ldots d_{n-4}]}{}^{I|M_2 \ldots M_{n-4}]}. \label{nu4}
\end{equation}
This expression provides a recurrence formula to obtain $\tilde{\nu}_{[d]}{}^{[M]}$, since it appears in both sides of the equation, but on the right-hand side one of its upper indices is contracted with one of its lower indices through $\tilde{\Pi}^a{}_I$. Equation~\eqref{nu4} allows us to derive the following expression:
\begin{eqnarray}
	&&\tilde{\Pi}^{a_1}{}_{I_1} \cdots \tilde{\Pi}^{a_k}{}_{I_k}\tilde{\nu}_{a_1 \ldots a_k d_{k+1} \ldots d_{n-4}}{}^{I_1 \ldots I_k M_{k+1} \ldots M_{n-4}}=\notag\\
	&&\frac{(n-4-k)}{k+4}\uac{\uac{h}}{}_{b[d_{k+1}|}\tilde{\Pi}^{b[M_{k+1}|}\tilde{\Pi}^{a_1}{}_{I_1} \cdots \tilde{\Pi}^{a_{k+1}}{}_{I_{k+1}}\tilde{\nu}_{a_1\ldots a_{k+1}| d_{k+2} \ldots d_{n-4}]}{}^{I_1 \ldots I_{k+1} |M_{k+2} \ldots M_{n-4}]},\notag\\
	\label{nu5}
\end{eqnarray}
which can be proven by induction ($k=0,1,\dots,n-5$). Using~\eqref{nu5} successively in the right-hand side of itself, we arrive at
\begin{eqnarray}
	&&\tilde{\Pi}^{a_1}{}_{I_1} \cdots \tilde{\Pi}^{a_k}{}_{I_k}\tilde{\nu}_{a_1 \ldots a_k d_{k+1} \ldots d_{n-4}}{}^{I_1\ldots I_k M_{k+1} \ldots M_{n-4}}=\notag\\
	&&\frac{(n-4-k)!}{(k+4)(k+5) \cdots (n-1)}\uac{\uac{h}}{}_{b_1[d_{k+1}|} \cdots \uac{\uac{h}}{}_{b_{n-4-k}|d_{n-4}]}\tilde{\Pi}^{b_1M_{k+1}} \cdots \tilde{\Pi}^{b_{n-4-k}M_{n-4}} \Theta,\label{nu6}
\end{eqnarray}
for $\Theta$ (of weight $n-3$) defined as
\begin{equation}
	\Theta:=\tilde{\Pi}^{a_1}{}_{I_1} \cdots \tilde{\Pi}^{a_{n-4}}{}_{I_{n-4}}\tilde{\nu}_{a_1\ldots a_{n-4}}{}^{I_1\ldots I_{n-4}}.\label{nu7}
\end{equation}
Setting $k=0$ in~\eqref{nu6}, we get
\begin{eqnarray}
	\tilde{\nu}_{[d]}{}^{[M]}=&&\frac{3!(n-4)!}{(n-1)!}\uac{\uac{h}}{}_{b_1[d_{1}|} \cdots \uac{\uac{h}}{}_{b_{n-4}|d_{n-4}]}\tilde{\Pi}^{b_1M_{1}} \cdots \tilde{\Pi}^{b_{n-4}M_{n-4}} \Theta.\label{nu8}
\end{eqnarray}
This expression is valid even for $n=4$. Using~\eqref{nu8} together with~\eqref{mPi} to rewrite the second row of~\eqref{B1}, gives
\begin{equation}
	\tilde{B}^{ab}{}_{IJ}=2\lambda^{ab}{}_c\tilde{\Pi}^c{}_{[I}m_{J]}+\sigma\uac{N}\tilde{\Pi}^{[a}{}_I\tilde{\Pi}^{b]}{}_J,\label{B2}
\end{equation} 
where we have defined $\uac{N}:=-3!(n-4)!\Theta/(n-1)!\sqrt{h}$. Notice that~\eqref{B2} is almost equal to~\eqref{BIJ6}, except for the prefactor of the terms including the product of $\tilde{\Pi}^a{}_{I}$ and $m_{J}$. However, we have not used~\eqref{constr32} yet. This will partially fix $\lambda^{ab}{}_c$ to make the previous equations equal.

From~\eqref{constr32}, we obtain the following expression for $\tilde{\nu}_{t\langle e\rangle}{}^{[M]}$:
\begin{equation}
	\tilde{\nu}_{t\langle e\rangle}{}^{[M]}=-\frac{1}{4!(n-4)!}\uac{\eta}_{tabcd\langle e\rangle}\epsilon^{IJKL[M]}\tilde{B}^{ab}{}_{IJ}\tilde{B}^{cd}{}_{KL}.\label{nu_2}
\end{equation}
Substituting~\eqref{B2} in the previous expression and plugging the result back into~\eqref{constr32}, we obtain, after the use of~\eqref{B2} and some algebra,
\begin{eqnarray}
	&&\uac{N}\epsilon^{[M]IJKL}m_I\tilde{\Pi}^e{}_J\left(-\lambda^{ab}{}_e\tilde{\Pi}^c{}_K\tilde{\Pi}^d{}_L-\lambda^{cd}{}_e\tilde{\Pi}^a{}_K\tilde{\Pi}^b{}_L+\lambda^{a[b}{}_e\tilde{\Pi}^c{}_K\tilde{\Pi}^{d]}{}_L+\lambda^{[cd}{}_e\tilde{\Pi}^{b]}{}_L\tilde{\Pi}^{a}{}_K\right)=0.\notag\\
\end{eqnarray}
Assuming $\uac{N}\neq 0$, multiplying the result by $\uac{\uac{h}}{}_{rf}\uac{\uac{h}}{}_{cm}\uac{\uac{h}}{}_{dn}\epsilon_{[M]PQRS}m^P\tilde{\Pi}^{fQ}\tilde{\Pi}^{mR}\tilde{\Pi}^{nS}$, and simplifying, we finally get
\begin{equation}
	\lambda^{ab}{}_r=N^{[a}\delta^{b]}_r,\label{lambda}
\end{equation}
for $N^a:=2\lambda^{ab}{}_b/(n-2)$. With this result, we observe that~\eqref{B2} exactly coincides with~\eqref{BIJ6}. In addition and as can be expected, we find $\tilde{\nu}_{[d]}{}^{[M]}=\sigma\tilde{\mu}_{[d]}{}^{[M]}$ and $\tilde{\nu}_{t\langle e\rangle}{}^{[M]}=\sigma\tilde{\mu}_{t\langle e\rangle}{}^{[M]}$, where the values of $\tilde{\mu}_{[d]}{}^{[M]}$ and $\tilde{\mu}_{t\langle e\rangle}{}^{[M]}$ are given in~\eqref{mu3} and~\eqref{mu4}, respectively. Therefore, the solution of~\eqref{constr12}-\eqref{constr32} involves the $n^2$ free functions $\tilde{\Pi}^{a}{}_{I}$, $\uac{N}$, and $N^a$ as in the previous section. Since the number of independent unknowns and the number of independent equations is the same as in the previous section too, the number of reducibility relations among the equations~\eqref{constr12}-\eqref{constr32} is the same as above, namely, $RR=\frac{1}{288}(n-4)(n-3)n^2(n+1)(n+2)(n^2-2n+9)$.

In conclusion, we have derived the same expressions~\eqref{SolPi2} and~\eqref{BIJ6} from the solution of the constraints~\eqref{constr12}-\eqref{constr32}. Substituting them into the action~\eqref{BFact4} leads to the same intermediate action~\eqref{action_Pi}, which produces the canonical formulation embodied in~\eqref{final_action}. Therefore, we have established that the manifestly $SO(n-1,1)$ [or, alternatively $SO(n)$] covariant canonical formulation~\eqref{final_action}, which was originally derived from the canonical analysis of the $n$-dimensional Palatini action without a cosmological constant, can also be derived from the $BF$-type actions~\eqref{BFact1} and~\eqref{BFact3}, something that in turn demonstrates the classical equivalence of both actions at the Hamiltonian level.

We remark that in order to arrive at the expressions~\eqref{SolPi2} and~\eqref{BIJ6}, the constraint~\eqref{constr1} can also be handled in a similar fashion as we did for the constraints of this section, although things get more complicated. Actually,~\eqref{constr12} can also be extracted from~\eqref{constr1}, but the fact that the latter mixes $\tilde{\Pi}^{aIJ}$, $\tilde{B}^{abIJ}$ and $\tilde{\mu}$ makes it more difficult to isolate~\eqref{constr12}.

\section{Conclusions}

In this paper we have performed, in a manifestly $SO(n-1,1)$ [or, alternatively $SO(n)$] covariant fashion, the canonical analysis of the formulation as a constrained $BF$ theory of general relativity in $n$-dimensions with a vanishing cosmological constant (we assume $n>4$, but our results hold in $n=4$ as well). Here we have considered the two formulations depicted in~\eqref{BFact1} and~\eqref{BFact3}, which differ in the form that the constraint on the $B$ field is imposed, but that nonetheless are classically equivalent. 

Our strategy is similar to that of Ref.~\cite{BFNoSCC}, where the case of $BF$ gravity in four dimensions with a cosmological constant and the Immirzi parameter is addressed (those results hold even if those two parameters are missing). Consequently, this work extends to higher dimensions the results of Ref.~\cite{BFNoSCC}. Our approach consisted in first performing the $(n-1)+1$ decompositions of the actions~\eqref{BFact1} and~\eqref{BFact3}, from which we read off the constraints imposed by the Lagrange multipliers of the corresponding formulation. We then proceeded to solve these constraints, which leads to a parametrization of all the components of the $B$ field, $\tilde{B}^{taIJ}\ (=\tilde{\Pi}^{aIJ}/2)$ and $\tilde{B}^{abIJ}$, in terms of the $n^2$ functions $\uac{N}$ (lapse), $N^a$ (shift), and $\tilde{\Pi}^{aI}$, as given by~\eqref{SolPi2} and~\eqref{BIJ6}. As established in this paper, both actions lead to the same parametrization, as expected. The substitution of these expressions in the original action then reduces it to~\eqref{action_Pi}, which constitutes an intermediate step in the canonical analysis of the $n$-dimensional Palatini action. By invoking the results of Ref.~\cite{PalatininD}, this produces the formulation~\eqref{final_action}, in which only first-class constraints are involved and the full gauge symmetry under $SO(n-1,1)$ or $SO(n)$ is kept manifest. We then end up with a phase space parametrized by the $2 n(n-1)$ variables $({\cal{Q}}_{aI},\tilde{\Pi}^{aI})$ subject to the $n(n+1)/2$ first-class constraints~\eqref{Gauss2}-\eqref{scalar2}, which leaves a total of $n(n-3)/2$ propagating degrees of freedom, the same as general relativity in $n$ dimensions. Notice that we have accomplished this canonical formulation without following the cumbersome and lengthy Dirac's approach that consists in enlarging phase space, introducing first-class and second-class constraints, and then explicitly solving the latter~\cite{dirac1964lectures}.

Although both actions~\eqref{BFact1} and~\eqref{BFact3} give rise to the same solution of the constraints and hence to the same canonical theory, we followed different approaches to deal with each case. In the case of the action~\eqref{BFact1} we obtained the two equations~\eqref{constr1}-\eqref{constr2}. It turns out that to find the required solution of the $B$ field, equation~\eqref{constr1} is enough, which indicates that the equation~\eqref{constr2} is superfluous and actually determines reducibility conditions for the whole system of equations. In constrast, for the action~\eqref{BFact3} we obtained the three equations~\eqref{constr12}-\eqref{constr32}, which were solved following a more algebraic approach than that for the first action. In this case, although there is reducibility among the equations, the three equations~\eqref{constr12}-\eqref{constr32} were used to find the proper solution of the $B$ field. It is worth mentioning that even though the issue of reducibility of the constraints is highly nontrivial in these $BF$-type formulations (see Ref.~\cite{Freid_Puz}), we did not need to deal with it in order to find the solution of the $B$ field.

\acknowledgments

We thank Diego Gonzalez for useful comments and suggestions on a preliminary version of this manuscript. This work was partially supported by Fondo SEP-Cinvestav and by Consejo Nacional de Ciencia y Tecnolog\'{i}a (CONACyT), M\'{e}xico, Grant No.~A1-S-7701. M.~C. gratefully acknowledges the support of a DGAPA-UNAM postdoctoral fellowship.

\appendix

\section{Alternative derivation of~\eqref{SolPi} and~\eqref{BIJ4}}\label{appendix}

The constraints~\eqref{constr2} are only present for $n>4$. Interestingly, they do not involve the components $\tilde{\Pi}^{aIJ}$. Instead of proceeding as in Sec.~\ref{secBF1}, where to solve the whole system of equations~\eqref{constr1} and~\eqref{constr2} it is enough to pay attention just to~\eqref{constr1}, we can start by first tackling~\eqref{constr2}. To solve it, we can follow the same strategy used in Sec.~\ref{secBF1}, separating the equation~\eqref{constr2} in different instances depending on the values of the internal indices as we did for~\eqref{constr1}. In this case, because $\tilde{B}^{abIJ}$ has two spatial indices and two internal indices, the procedure is exactly the same of Ref.~\cite{{Freid_Puz}}. Taking those results for granted, the solution for $\tilde{B}^{abIJ}$ is then given by~\eqref{BIJ4}, $\tilde{B}^{ab}{}_{IJ}=2\epsilon V^{[a}{}_{I}V^{b]}{}_{J}$, for some object $V^a{}_I$ of weight $+1/2$.

In order to fix $\tilde{\Pi}^{aIJ}$, we have to employ~\eqref{constr1}, which is equivalent to~\eqref{constr3}. Since we want an equation involving only $\tilde{\Pi}^{aIJ}$ and $\tilde{B}^{abIJ}$, we replace~\eqref{mu2} in~\eqref{V} to express $\tilde{\tilde{\mathcal{V}}}{}^{abc}{}_{IJKL}$ in terms of them. Substituting this result in~\eqref{constr3}, the latter becomes
\begin{eqnarray}
	&&2\tilde{\Pi}^{[a}{}_{IJ}\tilde{B}^{bc]}{}_{KL}+2\tilde{\Pi}^{[a}{}_{KL}\tilde{B}^{bc]}{}_{IJ}-\tilde{\Pi}^{[a}{}_{IL}\tilde{B}^{bc]}{}_{JK}\notag\\
	&&-\tilde{\Pi}^{[a}{}_{JK}\tilde{B}^{bc]}{}_{IL}-\tilde{\Pi}^{[a}{}_{IK}\tilde{B}^{bc]}{}_{LJ}-\tilde{\Pi}^{[a}{}_{LJ}\tilde{B}^{bc]}{}_{IK}=0.\label{constr4}
\end{eqnarray}

Let us define the symmetric object $\tilde{k}^{ab}:=V^{aI}V^b{}_I$ and assume that $k:=\det(\tilde{k}^{ab})$, $k\neq0$ (of weight $n-3$) is nonvanishing (otherwise, there would be nontrivial relations among the $V$'s). Likewise, we can introduce the orthogonal unit vector to $V^{aI}$ that fulfills $V^{aI}n_I=0$ and $n^In_I=\sigma'$, where $\sigma':=\sigma \, \text{sgn}(k)$; in terms of $V^{aI}$, this vector takes the form
\begin{equation}
	\label{n}
	n_{I}=\frac{1}{(n-1)! \sqrt{|k|}}\epsilon_{IJ_1 \ldots J_{n-1}}\underaccent{\tilde}{\eta}_{t a_1 \ldots a_{n-1}} V^{a_1 J_1}  \cdots V^{a_{n-1} J_{n-1}}
\end{equation}
and satisfies the identity $\uac{k}{}_{ab}V^{aI}V^b{}_J=\delta^I_J-\sigma' n^I n_J$, for $\uac{k}{}_{ab}$ the inverse of $\tilde{k}^{ab}$.

Substituting~\eqref{BIJ4} in~\eqref{constr4}, multiplying the result by $\uac{k}{}_{bd}\uac{k}{}_{ce}V^{dK}V^{eL}$, using the properties of $n^I$ and simplifying, we arrive at
\begin{eqnarray}
	\tilde{\Pi}^a{}_{IJ}=-\frac{2\sigma'}{n-2}n_{[I}\tilde{\Pi}^a{}_{J]K}n^K+\frac{2}{n-2}V^a{}_{[I}\lambda_{J]}-\frac{4\sigma'}{(n-1)(n-2)}\lambda_K n^KV^a{}_{[I}n_{J]},\label{Pi2}
\end{eqnarray}
where we have defined $\lambda_I:=\uac{k}{}_{ab}V^{aJ}\tilde{\Pi}^b{}_{JI}$ (of weight $+1/2$). Contracting~\eqref{Pi2} with $n^J$, we get $\tilde{\Pi}^a{}_{IJ}n^J=\lambda_J n^J V^a{}_{I}/(n-1)$, which plugged back in~\eqref{Pi2} leads to
\begin{equation}
	\tilde{\Pi}^a{}_{IJ}=2\epsilon V^a{}_{[I}M_{J]},\label{Pi3}
\end{equation}
with $M_I$ defined as
\begin{equation}
	M_I:=\frac{\epsilon}{n-2}\left(\lambda_I-\frac{\sigma'}{n-1}\lambda_J n^J n_I\right).
\end{equation}
Therefore, we have obtained~\eqref{SolPi}.


\bibliography{references}



\end{document}